\documentclass[12pt,a4paper]{article}
\usepackage[T1]{fontenc}
\usepackage[utf8]{inputenc}
\usepackage{float}
\usepackage{amsmath,amsfonts,amssymb,amsthm}
\usepackage{xurl}
\usepackage{graphicx}
\usepackage{epsfig}
\usepackage[outdir=figures/]{epstopdf}
\usepackage{subcaption}
\usepackage[left=3cm, right=3cm, top=3cm, bottom=3cm]{geometry}
\usepackage{algorithm}
\usepackage{algpseudocode}
\usepackage{bm}
\usepackage{tikz}
\usepackage{color}
\usepackage{mathrsfs}
\usepackage{cancel}
\usepackage[normalem]{ulem} 
\usepackage{eucal}
\usepackage{url}

\def\alphadot{\dot{\alpha}}
\def\deltadot{\dot{\delta}}
\def\rhodot{\dot{\rho}}
\def\Att{\mathcal{A}}

\def\angmom{\bm{c}}
\def\lenz{\bm{L}}
\def\energy{{\cal E}}

\def\erho{\bm{e}^{\rho}}

\def\eort{\bm{e}^\perp}
\def\erre{\bm{r}}
\def\erredot{\dot{\bm{r}}}
\def\qu{\bm{q}}  
\def\qudot{\dot{\bm{q}}} 

\def\DD{\bm{D}}
\def\EE{\bm{E}}
\def\FF{\bm{F}}
\def\GG{\bm{G}}
\def\JJ{\bm{J}}

\def\DeltaL{\bm{\Delta}_{\widetilde{\lenz}}}

\def\bzero{{\bf 0}}   

\def\C{\mathbb{C}}
\def\Z{\mathbb{Z}}
\def\P{\mathbb{P}} 

\newtheorem{proposition}{\bf Proposition}

\newtheorem{lemma}{\bf Lemma}
\newtheorem{definition}{\bf Definition}
\newtheorem{theorem}{\bf Theorem}

\title{Preliminary orbits with over-determined systems of Keplerian conservation laws}

\author{C. Grassi, G. F. Gronchi}

\begin{document}
\maketitle

\begin{abstract}
  We consider different choices of Keplerian conservation laws for the
  computation of preliminary orbits with two very short arcs (VSAs) of
  astrometric observations. In total we have 7 equations in 4
  unknowns.  Adding two auxiliary variables we can embed the full set
  of conservation laws into a polynomial system of 9 equations.  This
  complete system generically has no solutions.  However, combining
  these equations, in \cite{gbm15} the authors found an
  over-determined polynomial system that is consistent, and leads by
  variable elimination to a univariate polynomial $\mathsf{p}_9$ of
  degree 9 in one radial distance.  In \cite{gbm17} the authors showed
  that this corresponds to taking a subsystem with 7 equations of the
  complete system.  In this paper we consider all the other
  possibilities and we find two additional over-determined cases which
  are consistent and lead to a univariate polynomial $\mathsf{p}_{18}$
  of degree 18 in the same variable as $\mathsf{p}_9$. In the other
  over-determined cases the corresponding system is inconsistent. We
  also present a method to compute an approximate gcd of
  $\mathsf{p}_9$ and $\mathsf{p}_{18}$, that can allow us to find
  preliminary orbits that approximately satisfy inconsistent systems
  of conservation laws, or to discard incompatible pairs of VSAs.
  We show this through some numerical tests with real asteroid data.
  \end{abstract}


\section{Introduction}
\label{s:intro}

The modern asteroid surveys are producing very large databases of
optical observations.  Linking very short arcs (VSAs) of these
observations collected in different nights we can compute asteroid
orbits.  The search for efficient {\em linkage} algorithms is an
interesting mathematical problem: we show how it can be faced using
the first integrals of Kepler's problem.

The use of the Keplerian integrals (KI) for this purpose goes back to
the 70s, see \cite{th77}, where the authors proposed to employ
conservation of angular momentum and energy to compute preliminary
orbits for artificial Earth satellites. A Newton-Raphson method was
applied to compute the solutions but, according to \cite{taff84}, the
algorithm is extremely sensitive to the errors.  Later this algorithm
was improved, however the results were not satisfactory yet, see
\cite{trs84}.

More recently, in \cite{gdm10}, \cite{gfd11} the authors employed the
Keplerian integrals to write polynomial equations for the
problem: in \cite{gdm10} the same conservation laws as in \cite{th77}
were used, leading to a univariate polynomial equation of degree 48,
while in \cite{gfd11} the projection of the Laplace-Lenz vector onto a
suitable direction was used in place of the energy, leading to a
polynomial of degree 20.

These methods can be applied also to compute preliminary orbits of
satellites and space debris, and allow to include the $J_2$ effect
\cite{ftmr2010}.  A comparison between the two KI methods can be found
in \cite{dimare2011}, where these methods were successfully applied to
a large scale simulation of space debris observations in the upper
part of the LEO region.

In these methods the number of equations is equal to the number of
unknowns.  In \cite{gbm15} a combination of the Keplerian
conservation laws was used to obtain an over-determined polynomial
system which is consistent, that is it always admits solutions, at
least in the complex field. This system leads to a univariate
polynomial equation of degree 9.  In \cite{gbm17}, using Gr\"obner
basis theory, the authors gave an interpretation of the equations
considered in \cite{gbm15}.
 
Recently, the KI method introduced in \cite{gbm15} was successfully
used in \cite{rgbj24} as the first step of a complete orbit
determination pipeline, that has been applied to the observations of
the isolated tracklet file.\footnote{{\tt https://minorplanetcenter.net/iau/ITF/itf.txt.gz}} The authors were able to link several VSAs and compute a few
thousand orbits.

In Section~\ref{s:results} we introduce the mathematical setting for
this work and describe the results that we obtained in the attempt of
using as many Keplerian conservation laws as possible for the
computation of preliminary orbits.

\section{Preliminaries and description of the results}
\label{s:results}

Let us consider a reference frame with origin at the center of the
Sun. Denote by $\qu,\qudot$ the position and velocity vectors of the
observer and by $\erre,\erredot$ the position and velocity vectors of
the observed body at time $t$. We can write
\[
\begin{aligned}
  \erre &= \qu + \rho\erho,\\
  \erredot &=\qudot + \rhodot\erho + \rho\eort,
\end{aligned}
\]
where $\rho, \rhodot$ are the topocentric radial distance and velocity
of the observed body, $\erho$ is the line of sight unit vector and
$\eort = \frac{d\erho}{dt}$. Using angular coordinates
$(\alpha,\delta) \in [-\pi,\pi) \times [-\pi/2,\pi/2]$, which are
  usually right ascension and declination, we have
\[
\erho = (\cos\alpha\cos\delta, \sin\alpha\cos\delta,
\sin\delta).
\]
In this frame the Keplerian energy of the observed object is
\[
\energy = \frac{1}{2}|\erredot|^2 - \frac{\mu}{|\erre|},
\]
while the angular momentum $\angmom$ and the Laplace-Lenz vector
$\lenz$ are given by
\[
\begin{aligned}
  \angmom &= \erre\times\erredot,\\
  \mu\lenz &= \erredot\times\angmom - \frac{\mu}{|\erre|}\erre =
  \left(|\erredot|^2 - \frac{\mu}{|\erre|}\right)\erre -
  (\erredot\cdot\erre)\erredot.
\end{aligned}
\]
Given two VSAs of optical observations of a celestial body we can
compute two attributables
\[
  {\cal A}_1 = (\alpha_1,\delta_1,\alphadot_1,\deltadot_1),\qquad
  {\cal A}_2 = (\alpha_2,\delta_2,\alphadot_2,\deltadot_2)
\]
at epochs $t_1, t_2$, see \cite{milani2001}. Assuming that
$\qu_j,\qudot_j$ for $j=1,2$ are known, we try to determine an
initial orbit using the Keplerian conservation laws
\begin{equation}
  \angmom_1-\angmom_2 = \bzero,\qquad \mu(\lenz_1-\lenz_2) = \bzero,\qquad
  \energy_1-\energy_2 =0
  \label{syskepint}
  \end{equation}
in the unknowns $\rho_1,\rho_2,\rhodot_1,\rhodot_2$.
Introducing the auxiliary variables
\[
z_j = \frac{\mu}{|\erre_j|},\qquad j=1,2
\]
and the polynomials
\[
\begin{split}
  &\mu\widetilde\lenz_j = (|\erredot_j|^2 - z_j)\erre_j -
  (\erredot_j\cdot\erre_j)\erredot_j,\cr
  &\widetilde\energy_j = \frac{1}{2}|\erredot_j|^2 - z_j\cr
  &\zeta_j = z_j^2|\erre_j|^2-\mu^2,\cr
  \end{split}
\]
we embed \eqref{syskepint} into the polynomial system
\begin{equation}
  \angmom_1-\angmom_2 = \bzero,\quad
  \mu(\widetilde\lenz_1-\widetilde\lenz_2) = \bzero,\quad
  \widetilde\energy_1-\widetilde\energy_2 =0,\quad
  \zeta_1 = 0, \quad \zeta_2 = 0
  \label{syscomplete}
\end{equation}
where the unknowns are $\rho_1,\rho_2,\rhodot_1,\rhodot_2,z_1,z_2$.
However, system \eqref{syscomplete} is over-determined (9 equations
and 6 unknowns) and inconsistent, that is it generically does not
admit solutions, not even in the complex field, not even when ${\cal
  A}_1, {\cal A}_2$ belong to the same observed body, see Table
\ref{t:cardV} in Section \ref{s:diffgen}.

In Appendix \ref{app:AlgGeom} we define more precisely what we mean
when we claim that some property holds generically.

Let
\[
\DD_j = \qu_j\times\erho_j,\qquad j=1,2
\]
and assume
\[
\DD_1\times\DD_2\cdot\erho_1\times\erho_2 \neq 0,
\]
which generically holds.
In \cite{gbm17} the authors showed that the system
\begin{equation}
  \angmom_1-\angmom_2 = \bzero,\quad
  \mu(\widetilde\lenz_1-\widetilde\lenz_2) = \bzero,\quad
  \widetilde\energy_1-\widetilde\energy_2 =0,
\label{sysred9}
\end{equation}
obtained from \eqref{syscomplete} neglecting the last two
equations, is still over-determined but nevertheless consistent, i.e. it
generically admits solutions in the complex field, even when the two
VSAs do not belong to the same celestial body.
The proof has been conducted by computing a Gr\"obner basis of the ideal
$I$ generated by the polynomials
\begin{eqnarray}
\mathfrak{q}_1 &=& (\angmom_1-\angmom_2)\cdot\DD_1\times\DD_2, \label{q1}\\
\mathfrak{q}_2 &=& (\angmom_1-\angmom_2)\cdot\DD_1\times(\DD_1\times\DD_2), \label{q2}\\
\mathfrak{q}_3 &=& (\angmom_1-\angmom_2)\cdot\DD_2\times(\DD_1\times\DD_2), \label{q3}\\
\mathfrak{q}_4 &= &\mu(\widetilde{\lenz}_1-\widetilde{\lenz}_2)\cdot\erho_1\times\erho_2, \label{q4}\\
\mathfrak{q}_5 &= &\mu(\widetilde{\lenz}_1-\widetilde{\lenz}_2)\cdot\DD_2, \label{q5}\\
\mathfrak{q}_6 &= &\mu(\widetilde{\lenz}_1-\widetilde{\lenz}_2)\cdot\DD_1, \label{q6}\\
\mathfrak{q}_7 &= &\widetilde{\energy}_1-\widetilde{\energy}_2 \label{q7}
\label{Igen}
\end{eqnarray}
for the lexicographic (for short, lex) ordering with
\begin{equation}
  \rhodot_1 \succ \rhodot_2 \succ z_1 \succ z_2 \succ \rho_1 \succ \rho_2.
  \label{lexord}
\end{equation}
The selected projections of $\angmom_1-\angmom_2$ and
$\mu(\widetilde{\lenz}_1-\widetilde{\lenz}_2)$ have the purpose of
simplifying the polynomials in the system, in fact
\begin{eqnarray}
  \mathfrak{q}_1 &=& q(\rho_1,\rho_2), \label{qu}\\
  \mathfrak{q}_2 &=& |\DD_1\times\DD_2|^2\rhodot_1 - 
\JJ(\rho_1,\rho_2)\cdot \DD_1\times(\DD_1\times\DD_2), \label{rhodot1}\\
\mathfrak{q}_3 &=&  |\DD_1\times\DD_2|^2\rhodot_2 - 
\JJ(\rho_1,\rho_2)\cdot \DD_2\times(\DD_1\times\DD_2), \label{rhodot2}
\end{eqnarray}
where $q$ and $\JJ$ are defined as in \cite{gbm15} and are reported in
Appendix \ref{app:angmom} for completeness. Both $q$ and the three
components of $\JJ$ are quadratic forms in the variables
$\rho_1,\rho_2$, without the mixed monomial $\rho_1\rho_2$.

The projections of $\mu(\widetilde\lenz_1 - \widetilde\lenz_2)$ can be
written as follows:
\begin{eqnarray*}
  \mathfrak{q}_4 &=& -(\DD_1\cdot\erho_2)z_1 - (\DD_2\cdot\erho_1)z_2 + \mathfrak{f}_4,\\
\mathfrak{q}_5 &=& -(\DD_2\cdot\erre_1)z_1 + \mathfrak{f}_5,\\
\mathfrak{q}_6 &=& (\DD_1\cdot\erre_2)z_2 + \mathfrak{f}_6,
\end{eqnarray*}
for three polynomials $\mathfrak{f}_4$, $\mathfrak{f}_5$,
$\mathfrak{f}_6$ in the variables $\rho_1,\rho_2,\rhodot_1,\rhodot_2$
defined by
\[
\begin{split}
  \mathfrak{f}_4 &= \left[ |\erredot_1|^2\erre_1
    -(\erredot_1\cdot\erre_1)\erredot_1 - |\erredot_2|^2\erre_2 +
    (\erredot_2\cdot\erre_2)\erredot_2 \right]
  \cdot\erho_1\times\erho_2, \cr
  \mathfrak{f}_5 &= \left[ |\erredot_1|^2\erre_1 - (\erredot_1\cdot\erre_1)\erredot_1 + (\erredot_2\cdot\erre_2)\erredot_2 \right]\cdot \DD_2,  \cr
  %
%
  \mathfrak{f}_6 &= [-(\erredot_1\cdot\erre_1)\erredot_1 -
  |\erredot_2|^2\erre_2 +
  (\erredot_2\cdot\erre_2)\erredot_2]\cdot\DD_1. \cr
\end{split}
\]

We can obtain a Gr\"obner basis for $I$ of the form
\[
\{\rhodot_1+\mathfrak{g}_1, \ \rhodot_2+\mathfrak{g}_2, \
z_1+\mathfrak{g}_3, \ z_2+\mathfrak{g}_4, \ \rho_1 + \mathfrak{g}_5, \
\mathsf{p}_9\}
\]
with polynomials $\mathfrak{g}_k(\rho_2)$ and
$\mathsf{p}_9(\rho_2)$. Generically, we have\footnote{In \cite{gbm17},
\cite{gronchi_I-CELMECH} the form of the Gr\"obner basis is different:
to obtain the form displayed here we only need to divide some
polynomials by $\mathfrak{w}=\rho_1+\mathfrak{z}(\rho_2)$ and
$\mathfrak{v} = \mathfrak{g}_6(\rho_2)$.}
\[
\textrm{deg}(\mathfrak{g}_k)\leq 8 \quad (1\leq k\leq 5),\qquad
\textrm{deg}(\mathfrak{g}_6)= 9.
\]

From the properties of Gr\"obner bases, 9 is the minimal degree of a
univariate polynomial in $I$ in the unknown $\rho_2$.

\medbreak
Let us set
\begin{equation}
  \mathfrak{q}_8 = \zeta_1,\qquad \mathfrak{q}_9 = \zeta_2.
  \label{q8q9}
\end{equation}
If we add one of the latter polynomials to the generators of $I$, then
the corresponding ideal is the whole ring
$\C[z_1,z_2,\rhodot_1,\rhodot_2,\rho_1,\rho_2]$ and the related
variety is empty, see Section~\ref{s:diffgen}. In other words, for
example, there is generically no solution of the polynomial system
\[
\mathfrak{q}_k = 0,\qquad k=1,\ldots,8.
\]
Therefore, even if the two VSAs belong to the same celestial body, the
unavoidable astrometric errors, the errors in the numerical
computations, and the simplified dynamical model prevent us from
computing an orbit in this way.

\medbreak We can consider different ways of generating ideals with the
polynomials $\mathfrak{q}_k$, $k=1,\ldots,9$ whose variety is finite
and not empty.
This is done by dropping some of the $\mathfrak{q}_k$. We
will show in Section \ref{s:diffgen} that, with this procedure, the ideal
generated by $\mathfrak{q}_1,\ldots,\mathfrak{q}_7$ gives the smallest
non-empty variety (with 9 points), while the successive smallest
varieties have 18 points.  These are the varieties of the ideals
generated by
\[
\mathfrak{q}_k,\qquad  k=1,\ldots,6 
\]
and either $\mathfrak{q}_8$ or $\mathfrak{q}_9$.
The variety $V(J)$ of the polynomial ideal
\[
J = \langle \mathfrak{q}_1,\ldots,
\mathfrak{q}_6, \mathfrak{q}_9 \rangle
\]
(7 polynomials in 6 variables) is described in Section~\ref{s:VJ}.

In Section \ref{s:groebner} we compute a Gr\"obner basis of $J$ for
the lex ordering with relations \eqref{lexord}.
This basis has the form
\begin{equation}
\{\rhodot_1+\mathfrak{h}_1, \ \rhodot_2+\mathfrak{h}_2, \
z_1+\mathfrak{h}_3, \ z_2+\mathfrak{h}_4, \ \rho_1 + \mathfrak{h}_5, \
\mathsf{p}_{18}\}
\label{groebner18}
\end{equation}
with polynomials $\mathfrak{h}_k(\rho_2)$ and $\mathsf{p}_{18}$.
Generically we have
\[
\textrm{deg}(\mathfrak{h}_k)\leq 17 \quad (1\leq k\leq 5),\qquad
\textrm{deg}(\mathfrak{h}_6)= 18.
\]
Note that for a generic choice of the data $\mathcal{A}_j,\qu_j,\qudot_j
\,(j=1,2)$ we have
\begin{equation}
J = \langle \angmom_1 - \angmom_2,\ \mu(\widetilde\lenz_1 -
\widetilde\lenz_2),\ \zeta_2 \rangle.
\label{idealJ}
\end{equation}

When the two attributables belong to the
same celestial body we would like to be able to compute a common root
of these polynomials.  However, due to the already mentioned errors in
the observations, in the numerical computations and in the model, the
greatest common divisor (gcd) of $\mathsf{p}_9$ and $\mathsf{p}_{18}$
is 1, i.e. the two polynomials are relatively prime.
In Section~\ref{s:approxgcd} we show a way to compute orbits which are
defined starting from the roots $\rho_2$ of the approximate gcd of
$\mathsf{p}_9$ and $\mathsf{p}_{18}$. The singular value decomposition
is used as in \cite{Corless1995} to compute the degree of the
approximate gcd.
We conclude with a few numerical tests using real observations of
asteroids.

\section{Using different sets of Keplerian integrals}
\label{s:diffgen}

We consider the 9 polynomials
\begin{equation}
\mathfrak{q}_1,\ldots,\mathfrak{q}_9
\label{q1q9}
\end{equation}
introduced in Section \ref{s:results}.
We search for a subset of the $\mathfrak{q}_k$ polynomials generating
an ideal $I$ whose associated variety $V(I)$ is finite and not
empty. Moreover, we wish $V(I)$ to be as small as possible: for this
purpose we keep in the list of generators as many
$\mathfrak{q}_k$ polynomials as possible (provided
$V(I)\neq\emptyset$), and prefer the $\mathfrak{q}_k$ with low
degrees. The latter consideration leads us to always include
$\mathfrak{q}_1, \mathfrak{q}_2, \mathfrak{q}_3$, since they have
total degree 2. Moreover, $\mathfrak{q}_2$ and $\mathfrak{q}_3$ are
linear in $\rhodot_1$ and $\rhodot_2$ respectively, and allow a simple
elimination of these variables.

First we consider all the over-determined cases ($>6$ polynomials, $6$
unknowns) and then we study the relevant balanced cases ($6$ polynomials,
$6$ unknowns).

We compute the coefficients of the selected $\mathfrak{q}_k$ with the
values of the data given in \eqref{q_qdot_dat}, \eqref{adot_ddot_dat},
\eqref{sig_tau_dat} in Appendix~\ref{app:data}, and show in
Tables~\ref{t:cardV} and \ref{t:balanced} the number of points of the
varieties $V(I)$ obtained with the different choices of the generators.

\subsubsection*{Over-determined cases}

\begin{table}[t]
  \begin{center}
    \begin{tabular}{c|c|c|c|c|c|c|c}
      &\multicolumn{3}{c|}{\small$\mu(\widetilde{\lenz}_1-\widetilde{\lenz}_2)$}  &\multicolumn{1}{c|}{\small$\widetilde{\energy}_1-\widetilde{\energy}_2$} &$\zeta_1$ &$\zeta_2$ &    \cr
      \hline
      $\#$gen &$\mathfrak{q}_4$ &$\mathfrak{q}_5$ &$\mathfrak{q}_6$ &$\mathfrak{q}_7$ &$\mathfrak{q}_8$ &$\mathfrak{q}_9$ &$\#V(I)$\cr
      \hline
      9 &x &x &x &x &x &x &0\cr
      \hline\hline
      8 &x &x &x &x &{$\square$} &{$\square$} &0\cr
      \hline
      8  &x &x &x & &x &x &0\cr
      \hline
      8  &{\small$\triangle$} &{\small$\triangle$} &{\small$\triangle$} &x &x &x &0\cr    
      \hline\hline
      7 &x &x &x &x & & &9\cr
      \hline
      7  &x &x &x & &{$\square$} &{$\square$} &18\cr
      \hline
      7  &{$\triangle$} &{$\triangle$} &{$\triangle$} &x &{$\square$} &{$\square$} &0\cr
      \hline
      7  &{\small$\triangle$} &{\small$\triangle$} &{\small$\triangle$} & &x &x &0\cr    
      \hline
      7  &{o} &{o} &{o} &x &x &x &0\cr 
    \end{tabular}
  \end{center}
  \caption{The considered number of generators ($\#$gen) of the ideal
    $I$ is displayed in the first column. In each line we consider a
    different choice of generators: $\mathfrak{q}_1, \mathfrak{q}_2,
    \mathfrak{q}_3$ do not appear but are always included. The symbol
             {x} means that we keep that polynomial; instead, we keep
             only one of the three polynomials marked with {o}; for
             the triplets marked with {\small$\triangle$} we keep only
             two of them; for the pairs marked with {$\square$} we
             keep only one of the two. The last column gives the
             number of points in the variety associated to each
             considered ideal.}
  \label{t:cardV}
\end{table}

In Table~\ref{t:cardV} we consider all the possible choices of the
$\mathfrak{q}_k$ generators leading to an over-determined system.
We see that in the cases with 9 and 8 generators the corresponding
$V(I)$ is always empty, i.e. the considered systems are inconsistent.

On the other hand, among the ideals with 7 generators, we find the
case
\begin{equation}
I = \langle\mathfrak{q}_1,\ldots,\mathfrak{q}_7\rangle,
\label{9punti}
\end{equation}
with $\#V(I)=9$. This property actually holds for a generic choice of the data
(see \cite{gbm15}, \cite{gbm17}).
In Table~\ref{t:cardV} we also find the 2 cases
\[
I = \langle\mathfrak{q}_1,\ldots,\mathfrak{q}_6,\mathfrak{q}_8\rangle
\hskip 1cm
\mathrm{and}
\hskip 1cm
I = \langle\mathfrak{q}_1,\ldots,\mathfrak{q}_6,\mathfrak{q}_9\rangle,
\]
with $\#V(I)=18$. In Section~\ref{s:VJ} we will show that also this
property holds generically.
In the other cases of the table we find $\#V(I)=0$. More specifically,
we prove the following result.

\begin{proposition}
With reference to Table~\ref{t:cardV}, when $\#V(I)=0$ for the special
values \eqref{q_qdot_dat}, \eqref{adot_ddot_dat}, \eqref{sig_tau_dat},
then actually $\#V(I)=0$ for a generic choice of the data.
\end{proposition}

\begin{proof}
  First we note that, if we add an equation to an inconsistent system,
  the system remains inconsistent. For this reason we only need to
  show that the proposition holds in the cases of Table~\ref{t:cardV}
  with 7 generators.
  We have 12 cases of this type, listed in a synthetic way in
  Table~\ref{t:cardV}, that can be divided into 3
  groups:\footnote{Recall that $\mathfrak{q}_1, \mathfrak{q}_2,
  \mathfrak{q}_3$ are always included among the generators.}
  \begin{itemize}
  \item[i)] both $\mathfrak{q}_4$ and $\mathfrak{q}_7$ are generators
    of $I$:
    \[
    \begin{split}
      &\{\mathfrak{q}_4,\mathfrak{q}_5,\mathfrak{q}_7,\mathfrak{q}_8\},
      \ \{\mathfrak{q}_4,\mathfrak{q}_5,\mathfrak{q}_7,\mathfrak{q}_9\},\ \{\mathfrak{q}_4,\mathfrak{q}_6,\mathfrak{q}_7,\mathfrak{q}_8\},\cr
      &\{\mathfrak{q}_4,\mathfrak{q}_6,\mathfrak{q}_7,\mathfrak{q}_9\},\ \{\mathfrak{q}_4,\mathfrak{q}_7,\mathfrak{q}_8,\mathfrak{q}_9\};\cr
      \end{split}
    \]
    
  \item[ii)] only one between $\mathfrak{q}_4$ and $\mathfrak{q}_7$
    is a generator of $I$:
    \[
    \begin{split}
      &\{\mathfrak{q}_4,\mathfrak{q}_5,\mathfrak{q}_8,\mathfrak{q}_9\},\ 
      \{\mathfrak{q}_4,\mathfrak{q}_6,\mathfrak{q}_8,\mathfrak{q}_9\},\
      \{\mathfrak{q}_5,\mathfrak{q}_6,\mathfrak{q}_7,\mathfrak{q}_8\},\cr
      &\{\mathfrak{q}_5,\mathfrak{q}_6,\mathfrak{q}_7,\mathfrak{q}_9\},\
      \{\mathfrak{q}_5,\mathfrak{q}_7,\mathfrak{q}_8,\mathfrak{q}_9\},\
      \{\mathfrak{q}_6,\mathfrak{q}_7,\mathfrak{q}_8,\mathfrak{q}_9\};\cr
      \end{split}
    \]

  \item[iii)] neither $\mathfrak{q}_4$ nor $\mathfrak{q}_7$ is among
    the generators of $I$:
    \[
    \{\mathfrak{q}_5,\mathfrak{q}_6,\mathfrak{q}_8,\mathfrak{q}_9\}.
    \]
    
  \end{itemize}
  In group i), from $\mathfrak{q}_4=\mathfrak{q}_7=0$ we obtain
    \[
    z_1(\rho_1,\rho_2) = \frac{\mathfrak{f}_4 +
      (\DD_2\cdot\erho_1)\mathfrak{f}_7}{(\DD_1\cdot\erho_2) +
      (\DD_2\cdot\erho_1)},
    \qquad
    z_2(\rho_1,\rho_2) = \frac{\mathfrak{f}_4 -
      (\DD_1\cdot\erho_2)\mathfrak{f}_7}{(\DD_1\cdot\erho_2) +
      (\DD_2\cdot\erho_1)}.
    \]
    We can eliminate $z_1, z_2$ by substitution into the
    other two generators and obtain two polynomials $\mathfrak{p}_1,
    \mathfrak{p}_2$ in the variables $\rho_1,\rho_2$.
    
  \medbreak In group ii), using $\mathfrak{q}_4=0$ or
  $\mathfrak{q}_7=0$, we can eliminate one variable between $z_1$ and
  $z_2$. Assume we eliminate $z_2$, then we have three generators
  where $z_1$ appears. Computing the resultant of two pairs of these
  polynomials with respect to $z_1$ we obtain two polynomials
  $\mathfrak{p_1}, \mathfrak{p_2}$ in $\rho_1,\rho_2$.
  
  \medbreak
  In group  iii) we eliminate $z_1$ and $z_2$ by computing
  \[
  \mathfrak{p}_{1} = \mathrm{res}(\mathfrak{q}_5,\mathfrak{q}_8,z_1),
  \hskip 0.5cm
  \mathfrak{p}_{2} = \mathrm{res}(\mathfrak{q}_6,\mathfrak{q}_9,z_2).
  \]

Then, for all the cases we consider the ideal $K$ generated
by
\[
\{\mathfrak{q}_1,\ \mathfrak{p}_{1},\ \mathfrak{p}_{2}\}.
\]
In i) $K$ corresponds to the elimination ideal $K_\mathrm{elim} =
I\cap \C[\rho_1,\rho_2]$.  In ii) and iii) we have
$K\subseteq K_\mathrm{elim}$, so that
\[
V(K) \supseteq V(K_\mathrm{elim}).
\]

\medbreak To conclude the proof of the proposition we proceed in the
following way.  We homogenize the generators
$\mathfrak{q}_1,\ \mathfrak{p}_{1},\ \mathfrak{p}_{2}$ by adding a new
variable $x$, thus obtaining three homogeneous polynomials
\[
\widetilde{\mathfrak{q}}_1, \ \widetilde{\mathfrak{p}}_{1},
\ \widetilde{\mathfrak{p}}_{2}
\]
in the three variables $x,\rho_1,\rho_2$, and consider the homogeneous
ideal $\widetilde{K}$ generated by the latter
polynomials.\footnote{recall that an ideal $K$ is called homogeneous
if for every $p\in K$, each homogeneous term of $p$ belongs to $K$ as well. An
ideal generated by homogeneous polynomials is homogeneous, see
\cite{CLO2}.} Note that
\[
V(\widetilde{K})\supset V(K).
\]

For each of the cases above, we consider the data defined in
\eqref{q_qdot_dat}, \eqref{adot_ddot_dat}, \eqref{sig_tau_dat},
denoted shortly by $\bm{d}_*$, and compute a Gr\"obner basis for the
lex ordering with {\tt Maple 18} \cite{maple18}.
Note that Gr\"obner bases of an ideal with {\em specialized}
generators, i.e. generators whose symbolic coefficients are set to
some specific values, generically correspond to the specialization of
Gr\"obner bases computed with symbolic coefficients, see
\cite[Chapter 6.3]{CLO1}.
%
In all these cases we find a Gr\"obner basis fulfilling property ii)
of Theorem \ref{projweakNSS} in Appendix~\ref{app:AlgGeom}. Therefore
we obtain
\[
V(\widetilde{K}) = \emptyset.
\]
By Theorem \ref{t:Macaulay} in Appendix~\ref{app:AlgGeom}, the
Macaulay resultant $\mathrm{Res}(\bm{d})$ of
$\widetilde{\mathfrak{q}}_1, \widetilde{\mathfrak{p}}_1,
\widetilde{\mathfrak{p}}_2$, which is a polynomial function of the
data $\bm{d}\in{\cal D}$, is non-vanishing at $\bm{d}=\bm{d}_*$:
\[
\mathrm{Res}(\bm{d}_*)\neq 0.
\]
Therefore the Zariski-open set defined by
\[
\{ \bm{d}\in{\cal D}: \mathrm{Res}(\bm{d})\neq 0 \}
\]
is not empty, and Zariski-dense, see Appendix \ref{app:AlgGeom}.  We
conclude that $V(\widetilde{K})=\emptyset$ for a generic choice of the
data $\bm{d}$. This implies that $V(K)=\emptyset$, which yields
$V(K_\mathrm{elim})=\emptyset$, and therefore $V(I)=\emptyset$.

\end{proof}

\subsubsection*{Balanced cases}

\begin{table}[t]
  \begin{center}
    \begin{tabular}{c|c|c|c|c|c|c}
      \multicolumn{3}{c|}{\small$\mu(\widetilde{\lenz}_1-\widetilde{\lenz}_2)$}  &\multicolumn{1}{c|}{\small$\widetilde{\energy}_1-\widetilde{\energy}_2$} &$\zeta_1$ &$\zeta_2$ &    \cr
      \hline
      $\mathfrak{q}_4$ &$\mathfrak{q}_5$ &$\mathfrak{q}_6$ &$\mathfrak{q}_7$ &$\mathfrak{q}_8$ &$\mathfrak{q}_9$ &$\#V(I)$\cr
      \hline
        x & & &x &{$\square$} &{$\square$} &20\cr
      \hline
         &x & &x &x & &20\cr
      \hline
         & &x &x &x & &24\cr
      \hline
       &x & &x & &x &24\cr
      \hline
       & &x &x & &x &20\cr      
      \hline
       x              & & & &x &x &40\cr 
      \hline
       &{$\square$} &{$\square$} & &x &x &40\cr 
      \hline
       & & &x &x &x &48\cr 
    \end{tabular}
  \end{center}
  \caption{In each line we consider a different choice of generators
    of the ideal $I$: $\mathfrak{q}_1, \mathfrak{q}_2, \mathfrak{q}_3$
    do not appear, but are always kept in the set of generators. The
    symbols {x} and {$\square$} have the same meaning as in
    Table~\ref{t:cardV}.  The last column gives the number of points
    in the variety associated to each considered ideal.}
  \label{t:balanced}
\end{table}

In Table~\ref{t:balanced} we display the relevant cases with 6
generators, assuming the values $\bm{d}=\bm{d}_*$ of the data.  The
number of points belonging to $V(I)$, are computed with {\tt Maple
  18}.

We consider as non relevant (and we exclude from the list) the cases
that can be obtained by dropping polynomials from sets of generators
of a non-empty $V(I)$.  In fact, when we already have
$V(I)\neq\emptyset$, dropping a generator decreases the number of
Keplerian conservation laws that are satisfied, while our goal is to
use as many of them as possible.

\bigbreak
Looking at Tables~\ref{t:cardV}, \ref{t:balanced} we note that the
smallest non-empty variety
is made of 9 points and is associated to the ideal
$\langle\mathfrak{q}_1, \dots, \mathfrak{q}_7\rangle$ considered in
\cite{gbm15}.

The second-smallest varieties are the ones corresponding to the ideals
$\langle\mathfrak{q}_1, \ldots, \mathfrak{q}_6, \mathfrak{q}_8\rangle$
and $\langle\mathfrak{q}_1, \ldots, \mathfrak{q}_6,
\mathfrak{q}_9\rangle$. The latter, denoted by $V(J)$,
is discussed
in Section \ref{s:VJ}, but similar considerations apply to the
former. We chose to study $V(J)$ because all the other varieties in
Tables~\ref{t:cardV}, \ref{t:balanced} are either empty, or contain
more than 18 points.

\section{The algebraic variety $V(J)$}
\label{s:VJ}

We show that the variety $V(J)$ of the ideal $J$ defined in
\eqref{idealJ} is generically made of 18 points.

Using relations $\mathfrak{q}_2 = \mathfrak{q}_3 = 0$ we can easily
eliminate $\rhodot_1, \rhodot_2$ (see \eqref{rhodot1},
\eqref{rhodot2}) and consider the system
\begin{equation}
  q = 0,\quad
  \DeltaL = \bzero,\quad
  \zeta_2 = 0, 
  \label{modsysred18}
\end{equation}
where $q=\mathfrak{q}_1$ is defined by \eqref{qpoly}, and 
\[
\DeltaL = \mu(\widetilde{\lenz}_1 - \widetilde{\lenz}_2),
\]
where we have replaced $\rhodot_1, \rhodot_2$ with their expressions
in terms of $\rho_1, \rho_2$.

We use some ideas of the proofs of Theorem 1 and the related Lemmas in
\cite{gbm15} to show that system \eqref{modsysred18} generically
defines 18 points in the complex field.

\begin{lemma}
  Generically, there are no values of $z_1,z_2\in\C$ such that the point
  \[
  \widetilde C=(\rho_1'',\rho_2'',z_1,z_2),
  \]
fulfills system \eqref{modsysred18},  where $\rho_1'',\rho_2''$ are defined by
  \eqref{rhojsecondo}.

  \label{lem:deltaLC}
\end{lemma}
\begin{proof}
  We can write
  \[
  \mu\widetilde\lenz = \erredot\times\angmom-z\erre.
  \]
  Since $\angmom_1=\angmom_2=\bzero$ at $C=(\rho_1'',\rho_2'')$, see
  Lemma \ref{lem:pointC} in Appendix \ref{app:angmom}, at
  $\widetilde{C}$ we get
  \begin{equation}
    \DeltaL\cdot\DD_1 =
    z_2\erre_2''\cdot\DD_1 =
    z_2(\rho_2''-\rho_2')\erho_1\times\erho_2\cdot\qu_1,
    \label{DeltaL_D1}
  \end{equation}
  where $\erre_2'' = \rho_2''\erho_2 + \qu_2$.
  Since generically $(\rho_2''-\rho_2')\erho_1\times\erho_2\cdot\qu_1$
  is non-vanishing, we would get $z_2=0$, that is in contradiction with
  $\zeta_2=0$.

\end{proof}

\begin{lemma}
  Relation $\DeltaL = \bzero$ is generically equivalent to
  \begin{equation}
  \DeltaL\cdot\erho_1\times\erho_2 = 0\qquad\mbox{and}\qquad
  (\DeltaL\times\erre_1 = \bzero \quad\mbox{or}\quad
  \DeltaL\times\erre_2 = \bzero).
  \label{DeltaL_altern}
  \end{equation}
  \label{lem:DeltaL}
\end{lemma}
\begin{proof}
  Clearly $\DeltaL = \bzero$ implies \eqref{DeltaL_altern}.
  Viceversa, assume $\DeltaL\cdot\erho_1\times\erho_2 = 0$ and
  $\DeltaL\times\erre_1 = \bzero$. The second relation yields $\DeltaL
  = \bzero$ or $\DeltaL$ is nonzero and parallel to $\erre_1$, which
  is generically different from zero. The latter case leads to a
  contradiction because we would have
  \[
  0 =
  \erre_1\cdot\erho_1\times\erho_2 = \qu_1\cdot\erho_1\times\erho_2
  \]
  and the right-hand side is generically different from 0. Thus necessarily $\DeltaL = \bzero$.  The proof with $\DeltaL\times\erre_2
  = \bzero$ works in a similar way.
  
\end{proof}
  
\begin{lemma}
  Assume $q = \zeta_2 = \DeltaL\cdot\erho_1\times\erho_2 = 0$. Then
  $\DeltaL=\bzero$ is generically equivalent to
\begin{equation}
\left\{
\begin{array}{l}
  \DeltaL\cdot\DD_1 = 0,\cr
  \angmom\cdot \erho_1 \neq 0\cr
\end{array}
\right.
\hskip 1cm\mbox{or}\hskip 1cm
\left\{
\begin{array}{l}
  \DeltaL\cdot\DD_2 = 0,\cr
  \angmom\cdot \erho_2 \neq 0,\cr
\end{array}
\right.
\label{altern}
\end{equation}
where $\angmom = \angmom_1 = \angmom_2$ is the common value of the
angular momentum.
\label{lem:equiv}
\end{lemma}
\begin{proof}
  Let us assume that neither of the two systems in \eqref{altern} is
  satisfied and prove that $\DeltaL\neq\bm 0$. If
  \[
  \DeltaL\cdot\DD_1\neq 0\qquad \mbox{or}\qquad \DeltaL\cdot\DD_2\neq 0
  \]
  we immediately get the result. Otherwise, if $\DeltaL\cdot\DD_1 =
  \DeltaL\cdot\DD_2 = 0,$ our first assumption implies
  \begin{equation}
    \angmom\cdot\erho_1 = \angmom\cdot\erho_2 = \bm 0.
    \label{eq:cerho}
  \end{equation}
  From Lemma \ref{lem:pointC} in Appendix \ref{app:angmom}, we know
  that the point $C=(\rho_1'',\rho_2'')$ is the only solution of
  \eqref{eq:cerho}, and by Lemma \ref{lem:deltaLC} we cannot find
  values of $z_1,z_2\in\C$ such that the point
  $\widetilde{C}=(\rho_1'',\rho_2'',z_1,z_2)$ is a solution of
  \eqref{modsysred18}.
  
  On the other hand, assuming that $\DeltaL\ne\bm 0$, we can show that
  both systems in \eqref{altern} have no solution. Indeed, since
  $q=0$, we have $\angmom_1=\angmom_2$, so that the vectors $\DeltaL,
  \erre_1, \erre_2$ are coplanar and all orthogonal to $\angmom$.
  Therefore
  \[
  (\DeltaL\times\erre_j)\times\angmom = \bm 0\hskip 1cm j=1,2.
  \]
 By Lemma~\ref{lem:DeltaL}, since we are assuming
 $\DeltaL\cdot\erho_1\times\erho_2=0$ and $\DeltaL\ne\bm 0$,
 generically we have $\DeltaL\times\erre_j \neq \bzero$ for $j=1,2$.
 Thus we obtain, for some scalar functions $\kappa_j$,
    \begin{equation}
      \angmom\cdot\erho_{j} = \kappa_j \DeltaL\times\erre_{j}\cdot\erho_{j} =
      \kappa_j\DeltaL\cdot\erre_{j}\times\erho_{j} =
      \kappa_j\DeltaL\cdot\qu_{j}\times\erho_{j} = \kappa_j\DeltaL\cdot\DD_{j}.
    \label{eq:rel1}
  \end{equation}
  Therefore, both systems in \eqref{altern} are not satisfied.
  
\end{proof}


\begin{lemma}
  Generically, the polynomial system
\begin{equation}
q = \DeltaL\cdot\erho_1\times\erho_2 = \DeltaL\cdot\DD_1 = \zeta_2 = 0
\label{sys20}
\end{equation}
has 20 solutions in the complex field.
  \label{lem20}
\end{lemma}
\begin{proof} First note that in \eqref{sys20}
  the unknown $z_1$ appears only in equation
  $\DeltaL\cdot\erho_1\times\erho_2 = 0$, which gives only one value
  of $z_1$ once the values of the other variables are chosen.
  Eliminating $z_2$ from $\DeltaL\cdot\DD_1=\zeta_2=0$ we obtain
  \begin{equation}
    \mathscr{P}(\rho_1,\rho_2) = 0
    \label{Peq0},
  \end{equation}
  where\footnote{we are assuming $\rhodot_1, \rhodot_2$ have been
  eliminated in $\erredot_1, \erredot_2$.}
  \begin{equation}
    \begin{split}
    \mathscr{P}(\rho_1,\rho_2) &= |\erre_2|^2\left[
      -(\erredot_1\cdot\erre_1)(\erredot_1\cdot\DD_1) -
      |\erredot_2|^2(\erre_2\cdot\DD_1) +
      (\erredot_2\cdot\erre_2)(\erredot_2\cdot\DD_1) \right]^2\cr 
    & -\mu^2(\erre_2\cdot\DD_1)^2\cr
    \end{split}
    \label{defP} 
  \end{equation}
  is a polynomial with total degree 10 in $\rho_1, \rho_2$.
%
%
Using Bezout's theorem \cite{CLO1} with
equations $\mathscr{P} = 0$ and $q=0$,
we can state that system \eqref{sys20} has generically 20 solutions in
the complex field.

\end{proof}

\begin{lemma}
The point $P_1\equiv (\rho_1'',\rho_2')$ is a singular point
for the algebraic planar curve $\mathscr{P}(\rho_1,\rho_2) = 0$.
\label{lem:P1sing}
\end{lemma}
\begin{proof}
  We show that
  \begin{equation}
  \mathscr{P}(P_1) = \mathscr{P}_{\rho_1}(P_1) = \mathscr{P}_{\rho_2}(P_1) = 0,
  \label{P1sing}
  \end{equation}
  where the subscripts $\rho_1, \rho_2$ indicate derivatives with
  respect to these variables.

  First note that we can write \eqref{Peq0} as
\begin{equation}
|\erre_2|^2\left[ \DeltaL\cdot\DD_1 -
  z_2(\erre_2\cdot\DD_1) \right]^2 -
\mu^2(\erre_2\cdot\DD_1)^2 = 0,
\label{ls_equiv}
\end{equation}
where actually all the terms with $z_2$ cancel out.
From 
\begin{equation}
\erre_2\cdot\DD_1 = \rho_2 \qu_1\cdot\erho_1\times\erho_2 +
\erho_1\cdot\qu_1\times\qu_2 = (\rho_2 -
\rho_2')\erho_1\times\erho_2\cdot\qu_1
\label{r2_D1}
\end{equation}
we obtain that
\[
\erre_2\cdot\DD_1 = 0
\]
at $\rho_2=\rho_2'$.
Since $q(P_1) = 0$, see Appendix \ref{app:angmom}, we have
\begin{equation}
\DeltaL\cdot\angmom_1 = 0
\label{DeltaL_c}
\end{equation}
at $P_1$, whatever the value of $z_2$. From relation
\[
\begin{split}
\angmom_1\times\DD_1 &= \angmom_1\times(\erre_1\times\erho_1) =
(\angmom_1\cdot\erho_1)\erre_1 =
(\erredot_1\cdot\erho_1\times\erre_1)\erre_1\cr
& = (\erredot_1\cdot\erho_1\times\qu_1)\erre_1 = \bigl((\rho_1\eort_1
+ \qudot_1)\cdot\erho_1\times\qu_1\bigr)\erre_1 \cr
&= -(\rho_1-\rho_1'')(\erho_1\times\eort_1\cdot \qu_1)\erre_1
\end{split}
\]
we have that
\begin{equation}
  \angmom_1\times\DD_1=\bzero
  \label{cxD1null}
\end{equation}
at $P_1$.  Since $\angmom_1\neq \bzero$ at $P_1$, using
\eqref{DeltaL_c} and \eqref{cxD1null} we conclude that
relation
\[
\DeltaL\cdot\DD_1 = 0
\]
holds when we evaluate $(\rho_1,\rho_2)$ at $P_1$, whatever the value
of $z_2$.  In each term of the derivatives $\mathscr{P}_{\rho_1},
\mathscr{P}_{\rho_2}$ either $\erre_2\cdot\DD_1$ or
$\DeltaL\cdot\DD_1$ appears, and at $P_1$ both these terms and
therefore both derivatives vanish.

\end{proof}


  By Lemma \ref{lem:P1sing} the point $P_1=(\rho_1'',\rho_2')$ fulfills
  $\mathscr{P} = q = 0$.
For $\rho_2=\rho_2'$ the variable $z_2$ disappears in $\DeltaL\cdot\DD_1$. We use $\zeta_2 = \DeltaL\cdot\erho_1\times\erho_2=0$ to define
\[
z_2^\pm = \pm\frac{\mu}{\sqrt{|\erre_2(\rho_2')|}}
\]
and 
\[
z_1^\pm =
\frac{1}{\erho_2\cdot\DD_1}\left(\hat{\mathfrak{f}}_4(\rho_1'',
\rho_2') - (\erho_1\cdot\DD_2)z_2^\pm \right),
\]
where $\hat{\mathfrak{f}}_4$ is the function obtained by eliminating
$\rhodot_1, \rhodot_2$ in $\mathfrak{f}_4$.  With these definitions
the two points
\[
\widetilde{P}_1^\pm = (\rho_1'',\rho_2',z_1^\pm, z_2^\pm)
\]
are particular solutions of \eqref{sys20}.


\begin{proposition}
  Generically, system \eqref{modsysred18} has 18 solutions in the
  complex field, that is
  \[
  \#V(J)=18.
  \]
\label{p:18sol}
\end{proposition}
\begin{proof}
First we show that, assuming
$q=\DeltaL\cdot\erho_1\times\erho_2=\zeta_2=0$, the system on the left
in \eqref{altern}
has generically 18 solutions.  In fact,
we find here all the equations of system \eqref{sys20} plus relation
$\angmom\cdot\erho_1\neq 0$. By Lemma \ref{lem20}, system
\eqref{sys20} has 20 solutions. However, the two solutions
$\widetilde{P}_1^\pm = (\rho_1'',\rho_2',z_1^\pm, z_2^\pm)$ must be
discarded because of condition $\angmom\cdot\erho_1\neq 0$, which
excludes $P_1$. Generically, there are no other solutions with
$\rho_2=\rho_2'$ and we are left with 18 solutions. By Lemma
\ref{lem:equiv}, in the solutions of \eqref{modsysred18} we have to
take into account also the ones satisfying the system on the right in
\eqref{altern}, therefore we obtain that \eqref{modsysred18} has at
least 18 solutions.


Now we show that system \eqref{modsysred18} has at most 18
solutions.  We need to add equation $\DeltaL\cdot\DD_2 = 0$ to system
\eqref{sys20} that, as we have said, has 20 solutions.
To conclude, it is enough to note that
\[
\DeltaL(\widetilde{P}_1^\pm)\cdot\DD_2 \neq 0
\]
for some choice of the data.
In fact, relation $\DeltaL(\widetilde{P}_1^\pm)\cdot\DD_2=0$ defines a
closed set for the Zariski topology, whose complementary set is not
empty and therefore dense in the set of the data.


\end{proof}

\section{A Gr\"obner basis for $J$}
\label{s:groebner}

We consider the polynomial ring
\[
\C[\rho_1,\rho_2,\rhodot_1,\rhodot_2,z_1,z_2]
\]
with the lex monomial ordering, where
\[
\rhodot_1 \succ \rhodot_2 \succ z_1 \succ z_2 \succ \rho_1 \succ \rho_2,
\]
and the ideal
\[
\begin{split}
J &= \langle\angmom_1-\angmom_2, \ \mu(\widetilde{\lenz}_1 -
\widetilde{\lenz}_2), \ \zeta_2
\rangle
%
%
= \langle \mathfrak{q}_1, \ldots, \mathfrak{q}_6, \mathfrak{q}_9 \rangle.
\end{split} 
\]
We prove the following result.
\begin{proposition}
  For a generic choice of the data, we can find a Gr\"obner basis
  of the ideal $J$
of the form
\begin{equation*}
\{\rhodot_1+\mathfrak{h}_1, \ \rhodot_2+\mathfrak{h}_2, \
z_1+\mathfrak{h}_3, \ z_2+\mathfrak{h}_4, \ \rho_1 + \mathfrak{h}_5, \
\mathsf{p}_{18}\}
\end{equation*}
where $\mathfrak{h}_k = \mathfrak{h}_k(\rho_2)$, $k=1,\ldots,6$ and
$\mathsf{p}_{18}$ are univariate polynomials.  Generically we have
\[
\textrm{deg}(\mathfrak{h}_k)\leq 17 \quad (1\leq k\leq 5),\qquad
\textrm{deg}(\mathsf{p}_{18})= 18.
\]
  \label{p:groebner}
\end{proposition}
\begin{proof}
  Using \eqref{rhodot1} and \eqref{rhodot2} we can substitute the
  expressions $\rhodot_j = \rhodot_j(\rho_1,\rho_2)$ into
  $\mathfrak{q}_4, \mathfrak{q}_5, \mathfrak{q}_6$ and obtain the
  polynomials
  \begin{eqnarray*}
    \hat{\mathfrak{q}}_4 &=& -(\DD_1\cdot\erho_2)z_1 - (\DD_2\cdot\erho_1)z_2 + \hat{\mathfrak{f}}_4,\\
    \hat{\mathfrak{q}}_5 &=& -(\DD_2\cdot\erre_1)z_1 + \hat{\mathfrak{f}}_5,\\
    \hat{\mathfrak{q}}_6 &=& (\DD_1\cdot\erre_2)z_2 + \hat{\mathfrak{f}}_6,
  \end{eqnarray*}
  where $\hat{\mathfrak{f}}_j = \hat{\mathfrak{f}}_j(\rho_1,\rho_2)$
  for $j=4,5,6$. Thus, we can reduce to the elimination ideal
  \[
  K_1 = \langle \mathfrak{q}_1, \hat{\mathfrak{q}}_4, \hat{\mathfrak{q}}_5,
  \hat{\mathfrak{q}}_6, \mathfrak{q}_9 \rangle \subset \C[z_1,z_2,\rho_1,\rho_2].
  \]
  Since in the difference
      \[
|\erredot|^2\erre - (\erredot\cdot\erre)\erredot
\]
the term $\rhodot^2\rho\erho$ cancels out, using also $\erho\cdot\DD=0$
we can check that $\hat{\mathfrak{f}}_4, \hat{\mathfrak{f}}_5,
\hat{\mathfrak{f}}_6$ share the same Newton's polygon, which is given
in Figure \ref{newt1}(a).
  \begin{figure}
    \centering
    \includegraphics[width=\linewidth]{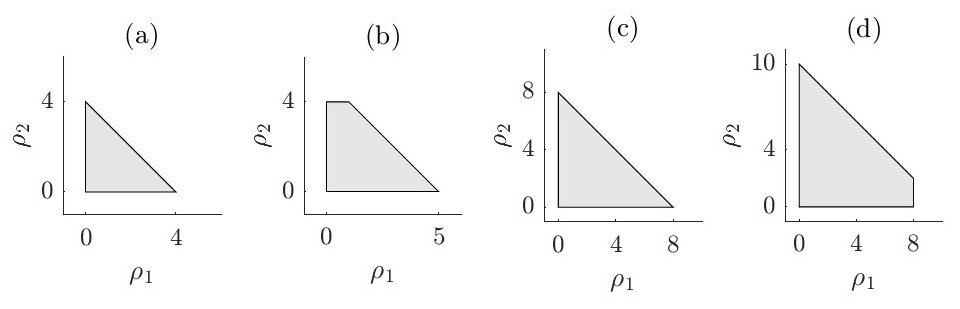}
    \caption{Newton's polygons of some polynomials in the proof. (a)
      $\hat{\mathfrak{f}}_4$, $\hat{\mathfrak{f}}_5$,
      $\hat{\mathfrak{f}}_6$. (b) $\mathfrak{w}_0$.  (c)
      $\mathfrak{w}_1$. (d) $\mathfrak{w}_2$, $\mathfrak{w}_3$.}
    \label{newt1}
  \end{figure}
  
  Consider then the division of $\hat{\mathfrak{q}}_5$ by
  $\hat{\mathfrak{q}}_4$ and let $\mathfrak{r}_0$
  be the remainder, that is obtained by inserting
  \[
  z_1 = -\frac{\DD_2\cdot\erho_1}{\DD_1\cdot\erho_2}z_2
  + \frac{\hat{\mathfrak{q}}_4}{\DD_1\cdot\erho_2}
  \]
  into $\hat{\mathfrak{q}}_5$.
 %
  Generically, we have
  \begin{equation}
  \mathfrak{r}_0 = z_2(A\rho_1+B) + \mathfrak{w}_0,
  \label{r0}
  \end{equation}
  for some coefficients $A,B$ depending only on the data and a
  polynomial $\mathfrak{w}_0(\rho_1,\rho_2)$ whose Newton's polygon is
  given in Figure \ref{newt1}(b).

  Substituting $\hat{\mathfrak{q}}_5$ with
  $\mathfrak{r}_0$
  in the
  generators of $K_1$, the variable $z_1$ appears only in
  $\hat{\mathfrak{q}}_4$ and we can restrict to the elimination
  ideal
  \[
  K_2 = \langle \mathfrak{q}_1, \mathfrak{r}_0, \hat{\mathfrak{q}}_6,
  \mathfrak{q}_9 \rangle \subset \C[z_2,\rho_1,\rho_2].
  \]
  \medbreak
  Next, we perform multi-polynomial division of $\mathfrak{q}_9$ by
  $\mathfrak{r}_0, \hat{\mathfrak{q}}_6$ and denote by $\mathfrak{r}_1$ the
  remainder. Generically, we have
  \begin{equation}
  \mathfrak{r}_1 = Cz_2^2 +Dz_2 + \mathfrak{w}_1
  \label{r1}
  \end{equation}
  for some coefficients $C,D$ and a polynomial $\mathfrak{w}_1 =
  \mathfrak{w}_1(\rho_1,\rho_2)$.  In fact, none of the terms of the
  remainder $\mathfrak{r}_1$ can be divisible by the leading monomial
  of $\mathfrak{r}_0$ or $\hat{\mathfrak{q}}_6$, and
  \[
  LM(\mathfrak{r}_0) = z_2\rho_1, \qquad LM(\hat{\mathfrak{q}}_6) =
  z_2\rho_2 .
  \] 
  To compute Newton's polygon of $\mathfrak{w}_1$, displayed in
  Figure~\ref{newt1}(c), we can proceed in this way: first we
  iteratively insert
    \begin{equation}
      z_2\rho_2 = -\frac{\DD_1\cdot\qu_2}{\DD_1\cdot\erho_2}z_2 -
      \frac{\hat{\mathfrak{f}}_6}{\DD_1\cdot\erho_2}
      \label{z2rho2}
\end{equation}
    into \[ \mathfrak{q}_9 = z_2^2(\rho_2^2 + 2\erho_2\cdot\qu_2\rho_2
    + |\qu_2|^2) - \mu^2
    \]
    and get
    \[
    Ez_2^2 + Fz_2\hat{\mathfrak{f}}_6 + G\hat{\mathfrak{f}}_6^2 + H
    \]
    for some constants $E,F,G,H$. Then transform $z_2\hat{\mathfrak{f}}_6$ by iteratively inserting either
    \eqref{z2rho2} or
    \[
z_2\rho_1 = -\frac{B}{A}z_2 - \frac{\mathfrak{w}_0}{A},
    \]
    which comes from \eqref{r0}, and check that Newton's polygon of
    $\mathfrak{w}_1$ is the same as $\hat{\mathfrak{f}}_6^2$.
 
  We add $\mathfrak{r}_1$ to the generators of
  $K_2$. Since $\mathfrak{r}_1$ corresponds to the remainder of the
  division between polynomials in $K_2$, adding $\mathfrak{r}_1$ does
  not make the ideal $K_2$ any larger.

\medbreak Let $\mathfrak{r}_2$ be the remainder of the division of
$\mathfrak{q}_9$ by $\mathfrak{r}_1$.
Using \eqref{r1}, we insert
\[
z_2^2 = -\frac{D}{C}z_2 - \frac{\mathfrak{w}_1}{C}
\]
into $\mathfrak{q}_9$ and find that
\[
\mathfrak{r}_2 = z_2(I\rho_2^2 + J\rho_2 + K) + \mathfrak{w}_2,
\]
for some constants $I,J,K$ and a polynomial
$\mathfrak{w}_2(\rho_1,\rho_2)$, whose Newton's polygon is shown in
Figure~\ref{newt1}(d).
We can substitute $\mathfrak{q}_9$ with $\mathfrak{r}_2$ in the set of
generators of $K_2$, so that
\[
K_2 = \langle \mathfrak{q}_1, \mathfrak{r}_0, \hat{\mathfrak{q}}_6,
\mathfrak{r}_1, \mathfrak{r}_2 \rangle \subset \C[z_2,\rho_1,\rho_2].
\]

Next, divide $\mathfrak{r}_2$ by $\hat{\mathfrak{q}}_6$ and let
$\mathfrak{r}_3$ be the remainder. We have
\[
\mathfrak{r}_3 = Lz_2 + \mathfrak{w}_3,
\]
for a constant $L$ and a polynomial $\mathfrak{w}_3(\rho_1,\rho_2)$
whose Newton's polygon is the same as $\mathfrak{w}_2$.
We substitute $\mathfrak{w}_3$ to $\mathfrak{r}_2$ in the set of
generators:
\[
K_2=\langle \mathfrak{q}_1, \mathfrak{r}_0,
\hat{\mathfrak{q}}_6, \mathfrak{r}_1, \mathfrak{r_3} \rangle \subset
\C[z_2,\rho_1,\rho_2].
\]
Dividing the polynomials $\mathfrak{r}_0, \hat{\mathfrak{q}}_6, 
\mathfrak{r}_1$ by $\mathfrak{r}_3$, and denoting by
$\mathfrak{r}_4, \mathfrak{r}_5, \mathfrak{r}_6$ the corresponding
remainders, that depend only on $\rho_1, \rho_2$, we obtain
\[
K_2 = \langle \mathfrak{q}_1, \mathfrak{r}_3, \mathfrak{r}_4,
\mathfrak{r}_5, \mathfrak{r}_6 \rangle \subset \C[z_2,\rho_1,\rho_2].
\]
Newton's polygons of $\mathfrak{r}_4, \mathfrak{r}_5,
  \mathfrak{r}_6$ are shown in Figures~\ref{newt2}(a),(c),(e).
    \begin{figure}
      \centering
      
              \includegraphics[width=\linewidth]{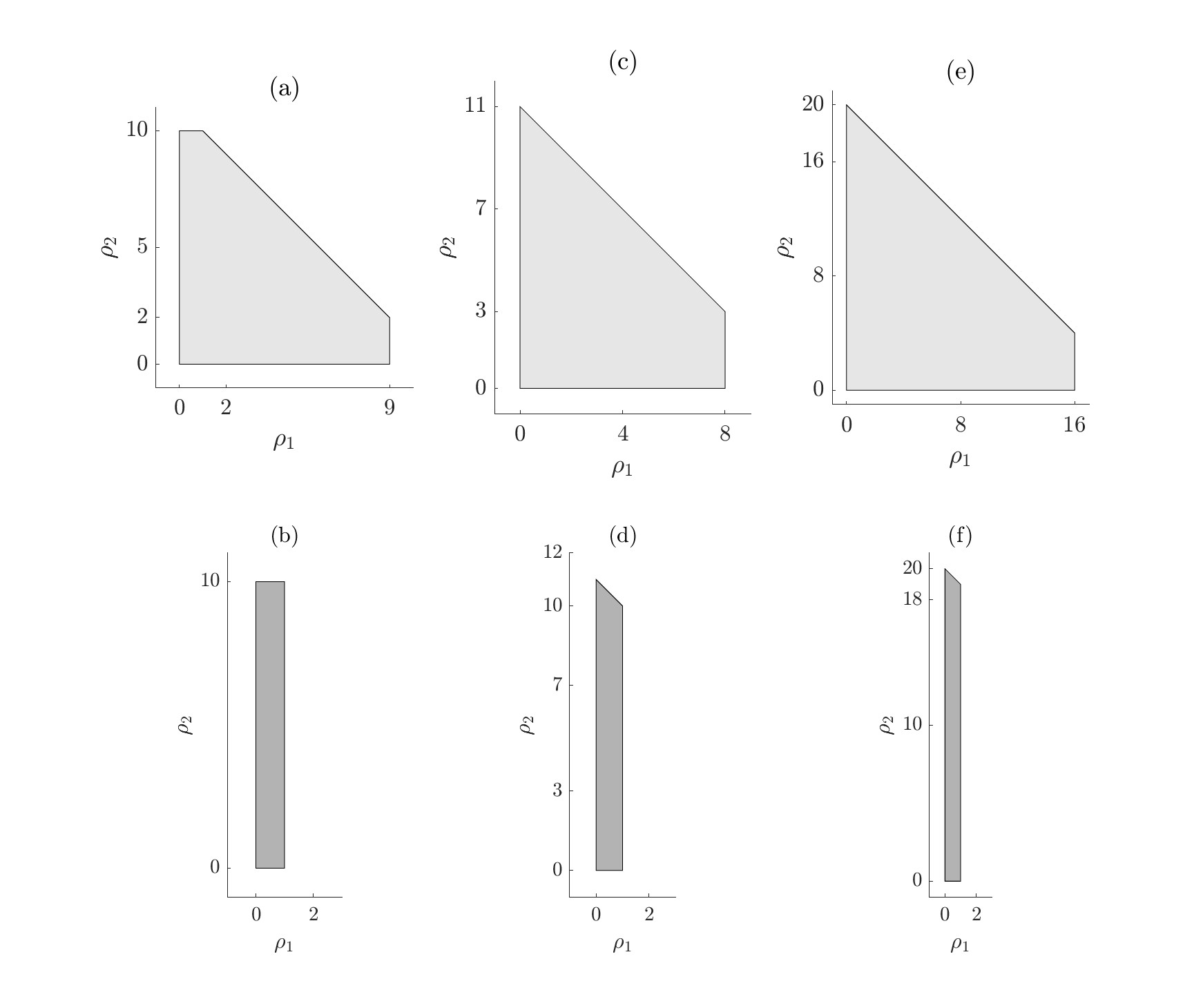}
      \caption{Newton's polygons of some polynomials in the proof. (a)
        $\mathfrak{r}_4$. (b) $\mathfrak{p}_1$. (c)
        $\mathfrak{r}_5$. (d) $\mathfrak{p}_2$. (e) $\mathfrak{r}_6$.
        (f) $\mathfrak{p}_3$.}
    \label{newt2}
  \end{figure}
  Note that $z_2$ appears only in $\mathfrak{r}_3$, thus we
  can
  restrict to the elimination ideal
  \[
  K_3 = \langle \mathfrak{q}_1, \mathfrak{r}_4, \mathfrak{r}_5,
  \mathfrak{r}_6 \rangle \subset \C[\rho_1,\rho_2].
  \]
  Now we divide $\mathfrak{r}_4, \mathfrak{r}_5, \mathfrak{r}_6$ by
  $\mathfrak{q}_1$ and call $\mathfrak{p}_1, \mathfrak{p}_2,
  \mathfrak{p}_3$ the remainders, so that
  \[
  K_3 = \langle \mathfrak{q}_1, \mathfrak{p}_1,
  \mathfrak{p}_2, \mathfrak{p}_3 \rangle.
  \]
  Generically, we get
\begin{equation}
\mathfrak{p}_k = \rho_1\xi_k(\rho_2) + \eta_k(\rho_2), \qquad k=1,2,3
\label{pk_details}
\end{equation}
      whose Newton's polygons are displayed in Figures~\ref{newt2}(b),(d),(f). 
\medbreak    
  Next, we consider a sequence of divisions
  \begin{equation}
    \begin{aligned}
      \mathfrak{p}_3 &= \mathcal{Q}_1\mathfrak{p}_1 + \mathfrak{s}_1,\\
      \mathfrak{p}_1 &= \mathcal{Q}_2\mathfrak{s}_1 + \mathfrak{s}_2,\\
      \mathfrak{s}_{k-1} &= \mathcal{Q}_{k+1}\mathfrak{s}_k + \mathfrak{s}_{k+1}, \qquad 2\leq k\leq 9\\
    \end{aligned}
    \label{eqn:serie_div}
  \end{equation}
where $\mathcal{Q}_{k} = \mathcal{Q}_{k}(\rho_2)$
with $\mathrm{deg}(\mathcal{Q}_1) = \mathrm{deg}(\xi_3) -
\mathrm{deg}(\xi_1), \ \mathrm{deg}(\mathcal{Q}_k) = 1 \
(k\geq 2)$,
that allows us to obtain polynomials
\[
\mathfrak{s}_i = \rho_1\mathfrak{u}_i(\rho_2) + \mathfrak{v}_i(\rho_2), \qquad 1\leq i\leq 10
\]
with
$\mathfrak{u}_i$, $\mathfrak{v}_i$ defined by
\[
\begin{array}{l}
  \xi_3 = \mathcal{Q}_1 \xi_1 + \mathfrak{u}_1,\cr
  \xi_1 = \mathcal{Q}_2 \mathfrak{u}_1 + \mathfrak{u}_2,\cr
  \mathfrak{u}_{k-1} = \mathcal{Q}_{k+1}\mathfrak{u}_{k} +
  \mathfrak{u}_{k+1}$, $2\leq k\leq 9, \cr
\end{array}
\hskip 1cm
\begin{array}{l}
  \eta_3 = \mathcal{Q}_1 \eta_1 + \mathfrak{v}_1,\cr
  \eta_1 = \mathcal{Q}_2 \mathfrak{v}_1 + \mathfrak{v}_2,\cr
  \mathfrak{v}_{k-1} = \mathcal{Q}_{k+1}\mathfrak{v}_{k} +
  \mathfrak{v}_{k+1}$, $2\leq k\leq 9. \cr
\end{array}
\]
Knowing that
\[
\mathrm{deg}(\mathfrak{u}_1)=9, \qquad \mathrm{deg}(\mathfrak{v}_1)=20,
\]
the relations above imply
\[
\begin{split}
  &\mathrm{deg}(\mathfrak{u}_{k+1}) = \mathrm{deg}(\mathfrak{u}_{k})-1, \hskip 1cm \mathrm{deg}(\mathfrak{u}_{10})=0,\cr
  &\mathrm{deg}(\mathfrak{v}_{k+1}) = \mathrm{deg}(\mathfrak{v}_{k})+1, \hskip 1cm  \mathrm{deg}(\mathfrak{v}_{10})=29,\cr
\end{split}
\]
so that
\[
\mathfrak{s}_{10} = M\rho_1 + \mathfrak{v}_{10}(\rho_2),
\]
for some constant $M$.

We selected $\mathfrak{p}_1,\mathfrak{p}_3$ to start the procedure
in \eqref{eqn:serie_div} because generically the two univariate
polynomials $\mathfrak{\xi}_1, \mathfrak{\xi}_3$ in \eqref{pk_details}
are relatively prime.
This claim is based on the computation of
$\mathrm{gcd}(\mathfrak{\xi}_1,\mathfrak{\xi}_3)$ with the data of
Appendix~\ref{app:data}. Since $\mathrm{gcd}(\mathfrak{\xi}_1,
\mathfrak{\xi}_3) = 1$ for the selected data, the resultant
$\mathrm{res}(\mathfrak{\xi}_1, \mathfrak{\xi}_3,\rho_2)$ evaluated at
these data is nonzero.
Now we write
\[
K_3 = \langle \mathfrak{q}_1, \mathfrak{p}_2, \mathfrak{s}_9,
\mathfrak{s}_{10}\rangle.
\]
Dividing $\mathfrak{q}_1, \mathfrak{p}_2, \mathfrak{s}_9$ by
$\mathfrak{s}_{10}$ we obtain remainders $\mathfrak{z}_{1},
\mathfrak{z}_{2}, \mathfrak{z}_{3}$, that are univariate polynomials
in $\rho_2$ with
\[
\mathrm{deg}(\mathfrak{z}_{1}) = 58,\qquad
\mathrm{deg}(\mathfrak{z}_{2}) = 38,\qquad
\mathrm{deg}(\mathfrak{z}_{3}) = 30.
\]



Recall that $V(J)$ has generically 18 points, as proved in Proposition
\ref{p:18sol}. From the construction made in this
section, the values of $\rhodot_1, \rhodot_2, z_1, z_2, \rho_1$ are
univocally determined by the value of $\rho_2$. Therefore, we
generically expect that all the univariate polynomials in $\rho_2$
contained in the ideal $J$ must be divisible by a polynomial of degree
18, that we call $\mathsf{p}_{18}$. Therefore, we can write
\[
K_3 = \langle \mathfrak{s}_{10}, \mathsf{p}_{18} \rangle.
\]

  Going back to the original ideal $J$, we can then construct the
  Gr\"obner basis
  \[
  \{\mathfrak{q}_2, \mathfrak{q}_3, \hat{\mathfrak{q}}_4, \mathfrak{r}_3,
  \mathfrak{s}_{10}, \mathsf{p}_{18}\},
  \]
  which can be easily transformed into a Gr\"obner basis of the form
  \eqref{groebner18}.

Evaluating the coefficients of the generators of $J$ at the data of
Appendix \ref{app:data}, we checked with {\tt Mathematica}
\cite{Mathematica} that the degrees of the polynomials considered in
this section are indeed generically attained.
  
\end{proof}


\section{Approximate gcd}
\label{s:approxgcd}

In Section~\ref{s:diffgen} we saw that including several conservation
laws can make the corresponding polynomial system inconsistent. In
case of inconsistency, even if the linkage is correct, we cannot find
any solution. On the other hand, using as many conservation laws as
possible can help to compute solutions with better accuracy, or to
discard a proposed linkage. For this reason we would like to be able
to deal also with inconsistent polynomial systems by searching for
approximate solutions.

We consider here the ideal generated by
\[
\{\mathfrak{q}_1,\ldots\mathfrak{q}_7,\mathfrak{q}_9\}
\]
which is generically empty, see Table~\ref{t:cardV}, therefore we
expect that the two univariate polynomials $\mathsf{p}_9$ and
$\mathsf{p}_{18}$, obtained by elimination from
\[
\langle \mathfrak{q}_1,\ldots\mathfrak{q}_7 \rangle, \hskip 1cm
\langle \mathfrak{q}_1,\ldots\mathfrak{q}_6,\mathfrak{q}_9 \rangle
\]
respectively, are relatively prime.
Then we try to compute compromise solutions of equations
\begin{equation}
\mathfrak{q}_1 = \ldots = \mathfrak{q}_7 = \mathfrak{q}_9 = 0
\label{q1q7q9}
\end{equation}
by searching for an approximate gcd of $\mathsf{p}_9$ and
$\mathsf{p}_{18}$, which provides us with compromise roots of these
two polynomials.  There are different methods to define and compute
approximate gcd of polynomial systems, see e.g. \cite{Boito,
  bini_boito, Corless1995}.  In this section we describe the procedure
that we have followed.

The first step is the computation of the coefficients of the
$\mathsf{p}_{18}$ polynomial. The system
\[
\begin{cases}
  \mathscr P(\rho_1, \rho_2) = 0\\
  q(\rho_1, \rho_2) = 0
\end{cases},
\]
with $\mathscr{P}$ defined as in \eqref{defP}, is similar to the one
derived in \cite{gfd11}: the only difference is that we use the
projection on $\DD_1$ instead of $\DD_2$. The resultant of
$\mathscr{P}$ and $q$ with respect to $\rho_1$ is a univariate
polynomial in $\rho_2$ of degree 20, that we denote by
$\mathsf{p}_{20}$. Then, using the result of Proposition
\ref{p:18sol}, the polynomial $\mathsf{p}_{18}$ can be written as
\begin{equation}
  \mathsf{p}_{18} =
  \displaystyle\frac{\mathsf{p}_{20}}{(\rho_2-\rho_2')^2}.
  \label{poly18}
\end{equation}
We can compute the coefficients of $\mathsf{p}_{18}$ by
evaluating the right-hand side of \eqref{poly18} at the $32^{nd}$
roots of unity $\omega_h = \exp(i\frac{h}{32})$, $h=0,\ldots,31$, by
the direct Fourier transform (DFT). The values of the coefficients of
$\mathsf{p}_{18}$ are given by its inverse (IDFT), like in
\cite{gdm10,gfd11}.

To get approximate solutions of \eqref{q1q7q9} we compute the
approximate gcd (or $\varepsilon$-gcd) of $\mathsf{p}_9$ and
$\mathsf{p}_{18}$.
For more details concerning the $\varepsilon$-gcd see \cite{bini_boito},
\cite{stetter}. We apply the procedure described in
\cite{Corless1995}, that we briefly recall here for completeness. Let
$m$ and $n$ be the degrees of the two polynomials, in our case $m=9$
and $n=18$. First, we consider the Sylvester matrix
$S\in\mathbb{R}^{(m+n)\times(m+n)}$ of $\mathsf{p}_9$ and
$\mathsf{p}_{18}$ and compute its singular value decomposition:
\[
S=U\Sigma V^T,
\]
where $U$ and $V$ are orthogonal matrices and $\Sigma$ is diagonal:
\[
\Sigma = \mathrm{diag}\{\sigma_1,\ldots,\sigma_{m+n}\},
\]
where the $\sigma_i$ are the singular values of $S$.

Suppose that $\sigma_1\ge\dots\ge\sigma_{m+n}$. Setting an error
tolerance $\varepsilon$, we find the maximum integer $k$ such that
\[
\sigma_k>\varepsilon\sqrt{m+n}\hskip 1cm\mathrm{and}\hskip 1cm
\sigma_{k+1}\le\varepsilon.
\]
If $k$ cannot be determined, i.e. if there is no such gap between two
consecutive singular values of $S$, the determination of the
$\varepsilon$-gcd fails. Otherwise, the numerical $\varepsilon$-rank
of $S$ is $k$ and we look for a polynomial
\[
d=\text{$\varepsilon$-gcd}{(\mathsf{p}_9,\mathsf{p}_{18})}
\]
of degree
\[
n_d=m+n-k.
\]
We compute $d$ by iterating the standard
Euclidean division algorithm for polynomials, like for the computation
of the usual gcd, but stopping the procedure when the degree of the
remainder is $n_d$.

Once the $\varepsilon$-gcd has been computed, we solve the polynomial
equation
\begin{equation}
  d(\rho_2)=0.
  \label{eq:gcd_roots}
\end{equation}
Let $\rho_2^{(i)}$, $i=1,\ldots,\ell$, be the real and positive solutions of
\eqref{eq:gcd_roots}. For each $\rho_2^{(i)}$ we have different ways
to compute the corresponding value of $\rho_1$. Indeed, we could use
for example
\begin{equation}
  \rho_1 + \mathfrak{g}_5(\rho_2)=0
  \label{eq:rho1g5}
\end{equation}
or
\[
  \rho_1 + \mathfrak{h}_5(\rho_2)=0.
\]
We chose instead to compute $\rho_1$ by solving
\begin{equation}
  q(\rho_1,\rho_2^{(i)}) = 0,
  \label{eq:qrho1}
\end{equation}
in order to obtain pairs $(\rho_1^{(i)},\rho_2^{(i)})$ that satisfy
$\angmom_1=\angmom_2$. However, equation \eqref{eq:qrho1} gives two
possible solutions $\rho_{1,1}^{(i)},\rho_{1,2}^{(i)}$. When both
solutions are real and positive we choose the one that minimizes
\[
|\rho_1 + \mathfrak{g}_2(\rho_2^{(i)})|
\]
Recall that
\[
\mathfrak{g}_5(\rho_2) = \alpha\tilde{a}_{1,0} + \beta\tilde{a}_{2,0},
\]
where $\tilde{a}_{1,0}, \tilde{a}_{2,0}$ are defined as in
Appendix~\ref{app:poly9}, and $\alpha, \beta$ are univariate polynomials
in $\rho_2$ such that
\[
\alpha\tilde{a}_{1,1} + \beta\tilde{a}_{2,1} = 1, 
\]
which exists because generically we have
\[
\mathrm{gcd}(\tilde{a}_{1,1},\tilde{a}_{2,1}) = 1,
\]
see \cite{gbm17}.
Then $\rhodot_1, \rhodot_2$ are found from $\mathfrak{q}_2 =
\mathfrak{q}_3 = 0$, see \eqref{rhodot1}, \eqref{rhodot2}, and a
preliminary orbit is computed.


\section{Numerical tests}

We present two test cases for the algorithm introduced in
Section~\ref{s:approxgcd}. In the first test we see how this method
can improve both preliminary orbits computed with $\mathsf{p}_9$ and
$\mathsf{p}_{18}$. In the second test we show how incompatible VSAs
can be discarded.

\begin{paragraph}{Test 1}
We consider the asteroid 2005 TF$_{181}$. From the set of
observations available for this object at
\url{https://newton.spacedys.com/astdys2/}, we select two tracklets,
one made of four observations taken on 16$^{th}$ April, 2024 from
the Pan-STARRS 2 observatory, the other consisting of four
observations taken on 26$^{th}$ May, 2024 from the Pan-STARRS 1
observatory and compute the corresponding two attributables $\mathcal
A_1, \mathcal A_2$ at times $t_1, t_2$. The real and positive roots of
the $\mathsf{p}_9$ polynomial in this case are
\[
\rho_2^{(1)}=0.496,\qquad \rho_2^{(2)}=2.396,\qquad
\rho_2^{(3)}=5.480,
\]
while those of $\mathsf{p}_{18}$ are
\[
\rho_2^{(1)}=0.873,\qquad \rho_2^{(2)}=2.722,\qquad
\rho_2^{(3)}=3.030.
\]
Applying the procedure described in Section
\ref{s:approxgcd}, with $\varepsilon=10^{-4}$, we obtain an $\varepsilon$-gcd
of $\mathsf{p}_9$ and $\mathsf{p}_{18}$ of degree
3, with the only real positive root
\[
\rho_2^* = 2.415.
\]
Computing the corresponding value of $\rho_1^*$ as described above we
get
\[
\rho_{1,1} = 2.253, \qquad \rho_{1,2} = -3.935.
\]

Selecting $\rho_1^* = \rho_{1,1}$ since it is the only positive
solution, we can compute $\dot\rho_2^*$ from
\[
\dot\rho_2 = \frac{(\bm
  J(\rho_1,\rho_2)\times\DD_1)\cdot(\DD_1\times\DD_2)}{|\DD_1\times\DD_2|^2},
\]
with $\bm J, \DD_1, \DD_2$ defined in Appendix \ref{app:angmom}.
Combining the values $\rho_2^*,\dot\rho_2^*$ with the attributable
$\mathcal A_2$, we can derive the orbital elements of the preliminary
orbit at time $t_2$. In Table \ref{t:orbits} we can see the orbital
elements of the orbits computed with the methods of \cite{gfd11} and
\cite{gbm15} and with the $\varepsilon$-gcd method. Furthermore, we
display the orbital elements of asteroid 2005 TF$_{181}$ provided by
the AstDyS catalog.

\begin{table}
  \begin{center}
  \begin{tabular}{c|c|c|c|c|c|c|c}
    &$a$ & $e$ & $i$ & $\Omega$ & $\omega$ & $\ell$ & $t$\\
    \hline
    $\mathsf{p}_{18}$ &3.027 & 0.034 & 9.299  & 26.626 & 38.919 & 217.856 & 60800 \\
    $\mathsf{p}_9$ &3.028 & 0.034 & 9.299  & 26.626 & 38.114 & 218.718 & 60800 \\
    $\varepsilon$-gcd &3.072 & 0.106 & 10.451 & 27.374 & 56.214 & 190.688 & 60800 \\
    known &3.063 & 0.135 & 10.759 & 27.618 & 58.215 & 185.612 & 60800
  \end{tabular}
  \end{center}
  \caption{Orbital elements of test case 2005 TF$_{181}$. On the first line
    the orbital elements computed with the method of \cite{gfd11}; on
    the second line those computed with the method of \cite{gbm15}; on
    the third line those obtained with the $\varepsilon$-gcd presented in
    this paper; on the fourth line the known orbital elements of the
    asteroid. Angles are expressed in degrees, times in MJD.}
  \label{t:orbits}
\end{table}
\end{paragraph}

\begin{paragraph}{Test 2}
To test the possibility of using the $\varepsilon$-gcd method to
filter out wrong linkages of tracklets, we consider two tracklets that
do not belong to the same asteroid. In particular, we take a first
tracklet belonging to asteroid 2005 TK$_{161}$, made of four
observations taken on 18$^{th}$ February, 2024 from Pan-STARRS 2, and
as second tracklet the first one used in the previous test, belonging
to asteroid 2005 TF$_{181}$. In this case,
there are two possible roots of $\mathsf{p}_9$:
\[
\rho_{2,1} = 1.654, \qquad \rho_{2,2} = 1.424.
\]
The solution $\rho_{2,2}$ would be selected as the best solution,
since it gives a lower value of the $\chi_4$ norm defined in
\cite{gbm15}. From this we would get the preliminary orbit 
\[
\begin{aligned}
  a &= 3.081, \qquad &e &= 0.235, \qquad &i &= 4.967,\\
  \Omega &= 21.250, \qquad &\omega &= 161.461, \qquad &\ell &= 18.362
\end{aligned}
\]
at time $t_2$.  Applying the method of Section \ref{s:approxgcd}, we
compute an $\varepsilon$-gcd between $\mathsf{p}_9$ and
$\mathsf{p}_{18}$ of degree 2, whose only real and positive root is
\[
\rho_2^* = 0.101.
\]
However, with this value of $\rho_2$, we are unable to get a
corresponding real and positive solution of
\eqref{eq:qrho1}. Therefore, having no acceptable values for $\rho_1$,
the linkage between these two tracklets with the $\varepsilon$-gcd
method rightly fails.
\end{paragraph}

\section{Acknowledgements}
This study was carried out within the Space It Up project funded by
the Italian Space Agency, ASI, and the Ministry of University and
Research, MUR, under contract n. 2024-5-E.0 - CUP n. I53D24000060005.

\section{Appendix}
\begin{appendix}

\section{Some results in Algebraic Geometry}
\label{app:AlgGeom}

We recall some definitions and results that has been used in this paper.
  
\begin{definition}
  Definition (5.6) in
  \cite[Chapter 3.5]{CLO2}: A property is
  said to {\em hold generically} for polynomials $f_1,\ldots,f_n$
  if there is a nonzero polynomial in
  the coefficients of the $f_i$ such that the property holds for all
  $f_1,\ldots f_n$ for which that polynomial is non-vanishing.
\label{def:generic}
\end{definition}
The definition above is based on Zariski topology, see
\cite{CLO1}. Recall that a non-empty Zariski open set is
Zariski-dense.

In our problem the coefficients of the polynomials $\mathfrak{q}_i$
depend on the data
\[
  \bm{d} = \{\Att_1,\Att_2,\qu_1,\qudot_1,\qu_2,\qudot_2\}.
  \]
  The dependence on $\alpha_j, \delta_j$ appear only through the sine
  or cosine of these angles. Taking as angular data these
  trigonometric functions instead of the angles themselves, the
  dependence of the coefficients on the data becomes polynomial.  In
  this way, the relation in the coefficients of Definition
  \ref{def:generic} is a polynomial relation $P(\bm{d})=0$ in the
  data.

  Endowing the set ${\cal D}$ of the data with the Zariski topology,
  if a property of the polynomials $f_1,\ldots,f_n$ fails exactly when
  $P(\bm{d})=0$ (a Zariski-closed set), and for a choice of the data,
  say $\bm{d}_*$, we have $P(\bm{d}_*)\neq 0$, then the open set
  \[
  \{\bm{d}\in{\cal D}: P(\bm{d})\neq 0\}
  \]
is not empty, and therefore is Zariski-dense.
  
We also recall the two following results, which are used in
Section~\ref{s:diffgen}.  \medbreak
\begin{theorem}
  {\em (Macaulay resultant, see \cite[Chapter 3.2]{CLO2})}. If
  $F_0,\ldots,F_n\in\C[x_0,\ldots,x_n]$
  with
  \[
  F_i = \sum_{|\alpha|=d_i}c_{i,\alpha}x^\alpha
  \]
  are homogeneous polynomials of positive degrees $d_0,\ldots,d_n$,
  with coefficients $c_{i,\alpha}$, then there is a unique polynomial
  $\mathrm{Res}\in \Z[c_{i,\alpha}]$ which has the following
  property:\\ the equations
  \[
  F_0 = \ldots = F_n = 0
  \]
  have a nontrivial solutions over $\C$ if and only if
  \[
  \mathrm{Res}(F_0,\ldots,F_n) = 0.
  \]
\label{t:Macaulay}
\end{theorem}

\begin{theorem}
  {\em (The Projective Weak Nullstellensatz, see \cite[Chapter
    8.3]{CLO1})}. Let $I$
  be a homogeneous ideal in $\C[x_0,\ldots,x_n]$. Then the following
  statements are equivalent:
\begin{itemize}
\item[i)] $V(I)\subset \P^n(\C)$ is empty; 
\item[ii)] Let $G$ be a reduced Gr\"obner basis for $I$ with respect to
  some monomial ordering. Then for each variable $x_i\ (i=0,\ldots,n)$
  there is $g\in G$ such that $LT(g)$ is a non negative power of $x_i$.
\end{itemize}
\label{projweakNSS}
\end{theorem}

\section{Angular Momentum}
\label{app:angmom}

The angular momentum vector $\angmom_i$ at time $t_i\ (i=1,2)$ can be written
as
\[
\angmom_i = \DD_i\rhodot_i + \EE_i\rho_i^2 + \FF_i\rho_i + \GG_i,
\]
with
\[
\DD_i = \qu_i\times\erho_i,\quad \EE_i = \erho_i\times\eort_i,\quad
\FF_i = \qu_i\times\eort_i + \erho_i\times\qudot_i,\quad \GG_i =
\qu_i\times\qudot_i.
\]
We write the difference of the angular momentum vectors as
\begin{equation}
  \angmom_2-\angmom_1 = \bm J(\rho_1,\rho_2) + \DD_2\rhodot_2 -
  \DD_1\rhodot_1,
  \label{eqn:diffc}
\end{equation}
with
\begin{equation}
  \bm J(\rho_1,\rho_2) = \EE_2\rho_2^2-\EE_1\rho_1^2 + \FF_2\rho_2 -
  \FF_1\rho_1 + \GG_2 - \GG_1.
\end{equation}
Projecting \eqref{eqn:diffc} onto $\DD_1\times\DD_2$, we can eliminate
the variables $\rhodot_1, \rhodot_2$ and obtain the polynomial
\begin{equation}
  q(\rho_1,\rho_2) =  q_{2,0}\rho_1^2 + q_{1,0}\rho_1 +
  q_{0,2}\rho_2^2 + q_{0,1}\rho_2 + q_{0,0},
  \label{qpoly}
\end{equation}
with
\[
\begin{array}{l}
q_{2,0} = -\EE_1\cdot\DD_1\times\DD_2,\cr
q_{1,0} = -\FF_1\cdot\DD_1\times\DD_2,\cr
\end{array}
\hskip 1cm
\begin{array}{l}
q_{0,2} = \EE_2\cdot\DD_1\times\DD_2,\cr
q_{0,1} = \FF_2\cdot\DD_1\times\DD_2.\cr
\end{array}
\]

We consider the projections of the angular momentum vectors onto the
line of sight vectors:
\[
c_{ij} = \angmom_i\cdot\erho_j, \qquad i,j=1,2.
\]
The equations
\[
c_{11}(\rho_1,\rho_2) = 0, \qquad c_{22}(\rho_1,\rho_2) = 0
\]
define straight lines in the $(\rho_1,\rho_2)$ plane, while
\[
c_{12}(\rho_1,\rho_2) = 0, \qquad c_{21}(\rho_1,\rho_2) = 0
\]
define conic sections, see Figure~\ref{fig:cij}.

\begin{figure}[h!]
  \begin{center}
\begin{tikzpicture}
  \coordinate (O) at (0,0);
  \coordinate (xl) at (-1,0);
  \coordinate (xr) at (4.5,0);
  \coordinate (yd) at (0,-1.5);
  \coordinate (yu) at (0,3);
  
  \draw[-latex] (xl)--(xr) node [below] {$\rho_1$};
  \draw[-latex] (yd)--(yu) node [left] {$\rho_2$};

  \coordinate (c11d) at (3.07,-1.5);
  \coordinate (c11u) at (3.07,3);
  \coordinate (c22l) at (-1,1.57);
  \coordinate (c22r) at (4.5,1.57);
  \draw[dashed] (c11d)--(c11u);
  \draw[dashed] (c22l)--(c22r);
  \draw (c11u) node [right] {$c_{11}=0$};
  \draw (c22r) node [above] {$c_{22}=0$};

  \draw (1.8,2.9) node [above,color=black!60!white] {$c_{12}=0$};
  \draw (4.1,1) node [below,color=black!90!white] {$q=0$};
  \draw (1.5,-1.1) node [below,color=black!60!white] {$c_{21}=0$};
  
  \draw[scale=2.,color=black!90!white] (0.9,0.3) ellipse (25pt and 20pt); 
  \draw[scale=2.,color=black!60!white,dashed] (0.9,1.02) ellipse (22pt and 12pt); 
  \draw[scale=2.,color=black!60!white,dashed] (0.715,0.3) ellipse (28pt and 25pt); 

  \coordinate (C) at (3.07,1.57);
  \draw (C) node [right] {$C$};
  \fill (C) circle (0.3mm);
  
  \coordinate (P1) at (3.07,-0.37);
  \draw (P1) node [right] {$P_1$};
  \fill (P1) circle (0.3mm);
  \draw[dotted] (-1,-0.37)--(4,-0.37);

  \draw (0.1,-0.3) node[left] {$\rho_2'$};
  \draw (0.1,1.75) node[left] {$\rho_2''$};
  
  \coordinate (P2) at (0.53,1.57);
  \draw (P2) node [above] {$P_2$};
  \fill (P2) circle (0.3mm);
  \draw[dotted] (0.53,-1.5)--(0.53,3);
  
  \draw (0.53,0.1) node[below] {$\rho_1'$};
  \draw (2.9,0.1) node[below] {$\rho_1''$};

\end{tikzpicture}
\end{center}
  \caption{Curves given by $q=0$, $c_{ij}=0$.}
  \label{fig:cij}
\end{figure}
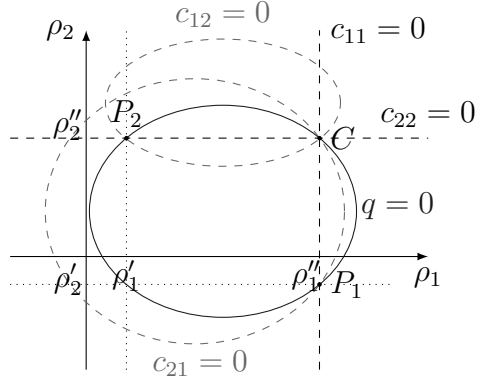

The lines given by $c_{11}=0$ and $c_{22}=0$ intersect in a point
$C=(\rho_1'',\rho_2'')$, with
\begin{equation}
\rho_1'' =
\frac{\qu_1\times\qudot_1\cdot\erho_1}{\erho_1\times\eort_1\cdot\qu_1},
\qquad \rho_2'' =
\frac{\qu_2\times\qudot_2\cdot\erho_2}{\erho_2\times\eort_2\cdot\qu_2}.
\label{rhojsecondo}
\end{equation}
In \cite{gbm15} the authors proved the following result:
\begin{lemma}
  If $q_{20},q_{02}\ne0$, which is generically true, $C$ satisfies the
  equation $q=0$ and is the only point where both $\angmom_1$ and
  $\angmom_2$ vanish.
  \label{lem:pointC}
\end{lemma}
Each straight line $c_{ii}=0$, for $i=1,2$, intersects the conic
$q=0$ in another point $P_i$. These points are given by
\[
P_1 = (\rho_1'',\rho_2'), \qquad P_2 = (\rho_1',\rho_2''),
\]
where $\rho_1''$ and $\rho_2''$ are defined in \eqref{rhojsecondo},
and
\begin{equation}
\rho_1' = \frac{\qu_1\times\qu_2\cdot\erho_2}{\erho_1\times\erho_2\cdot\qu_2},
\qquad
\rho_2' = \frac{\qu_1\times\qu_2\cdot\erho_1}{\erho_1\times\erho_2\cdot\qu_1}.
\label{rhojprimo}
\end{equation}
In particular, we note that also the points $P_i$ belong to the conic
$q=0$ and, generically, they are both different from $C$, therefore
at $P_i$ the angular momentum is not zero.

\section{The polynomial $\mathsf{p}_9$}
\label{app:poly9}

In \cite{gbm17}, \cite{gronchi_I-CELMECH} the authors derived a
univariate polynomial of degree 9 in $\rho_2$, that we denoted by
$\mathsf{p}_9$, which is one of the elements of a Gr\"obner basis of
the ideal generated by $\mathfrak{q}_1,\dots,\mathfrak{q}_7$ for the
lex ordering with
\[
\rhodot_1 \succ \rhodot_2 \succ z_1 \succ z_2 \succ \rho_1 \succ \rho_2.
\]

We briefly describe how $\mathsf{p}_9$ is constructed. In \cite{gbm15}
the authors reduced the problem to the computation of the solutions of
the over-determined system
\begin{equation}
q = 0, \hskip 1cm
\boldsymbol\xi = \bzero
\label{qxi}
\end{equation}
in the variables $\rho_1, \rho_2$, where $q$ is defined as in
\eqref{qpoly} and
\[
\begin{split}
\boldsymbol\xi &= \left[\mu(\widetilde{\lenz}_1-\widetilde{\lenz}_2) - (\widetilde{\energy}_1\erre_1 -
  \widetilde{\energy}_2\erre_2)\right]\times(\erre_1-\erre_2)\\
%
%
& = \left[\mu(\widetilde{\lenz}_1-\widetilde{\lenz}_2) -
  (\widetilde{\energy}_1-\widetilde{\energy}_2)\erre_1
  \right]\times(\erre_1-\erre_2)\\
& =
\mu(\widetilde{\lenz}_1-\widetilde{\lenz}_2)\times(\erre_1-\erre_2) +
(\widetilde{\energy}_1-\widetilde{\energy}_2)\erre_1\times\erre_2.
\end{split}
\]
In the expression of $\boldsymbol{\xi}$ the variables $z_1, z_2$
cancel out, and we have eliminated $\rhodot_1, \rhodot_2$ using
$\mathfrak{q}_2=\mathfrak{q}_3=0$.

System \eqref{qxi} is generically equivalent to
\[
q = p_1 = p_2 = 0,
\]
where $p_1, p_2$ are the projections of $\boldsymbol\xi$ onto
$\erho_1, \erho_2$ respectively, and can be written as
\[
p_1(\rho_1,\rho_2) = \sum_{h=0}^4a_{1,h}(\rho_2)\rho_1^h, \qquad
p_2(\rho_1,\rho_2) = \sum_{h=0}^5a_{2,h}(\rho_2)\rho_1^h,
\]
for some polynomials $a_{i,j}$ in $\rho_2$. Moreover, $q$
can be written as
\[
q(\rho_1,\rho_2) = \sum_{h=0}^2b_h(\rho_2)\rho_1^h,
\]
where
\[
b_0(\rho_2) = q_{0,2}\rho_2^2 + q_{0,1}\rho_2 + q_{0,0}, \qquad b_1 =
q_{1,0},\qquad b_2=q_{2,0}.
\]
Assuming $b_2\ne0$, which generically holds, we can define
\begin{equation}
  \begin{split}
    &\beta_1=1,\qquad \beta_2 = -\frac{b_1}{b_2}, \qquad
    \gamma_2=-\frac{b_0}{b_2},\\
    &\beta_{h+1}=\beta_h\beta_2+\gamma_h,\qquad
    \gamma_{h+1}=\beta_h\gamma_2,\qquad h=2,3,4,
  \end{split}
  \label{betagamma}
\end{equation}
and
\[
\begin{split}
&\tilde{a}_{1,1} = a_{1,1} + \sum_{h=2}^4a_{1,h}\beta_h, \qquad
  \tilde{a}_{1,0} = a_{1,0} + \sum_{h=2}^4a_{1,h}\gamma_h, \qquad\\
  &\tilde{a}_{2,1} = a_{2,1} + \sum_{h=2}^5a_{2,h}\beta_h, \qquad
  \tilde{a}_{2,0} = a_{2,0} + \sum_{h=2}^5a_{2,h}\gamma_h.
\end{split}
\]
Finally, define
\[
\tilde{p}_1 = \tilde{a}_{1,1}(\rho_2)\rho_1 + \tilde{a}_{1,0}(\rho_2), \qquad
\tilde{p}_2 = \tilde{a}_{2,1}(\rho_2)\rho_1 + \tilde{a}_{2,0}(\rho_2).
\]
With this notation, the expression of $\mathsf{p}_9$
is
\begin{equation}
  \mathsf{p}_9 = \mathrm{res}(\tilde{p}_1,\tilde{p}_2,\rho_1) =
  \tilde{a}_{1,1}\tilde{a}_{2,0} - \tilde{a}_{1,0}\tilde{a}_{2,1}.
  \label{poly9}
\end{equation}

\section{Selected data with rational values}
\label{app:data}

We list below the data that we have used in Section~\ref{s:diffgen}:
working with rational numbers allows us to make exact computations.
In particular, we assume $\mu=1$ and consider the following observer
positions and velocities:
\begin{equation}
\qu_1 = (1,0,0),\quad \qudot_1 = (1,1,1/2),\quad \qu_2 = (0,1,0),\quad
\qudot_2=(-1,2,-1).
\label{q_qdot_dat}
\end{equation}
The chosen angular rates are
\begin{equation}
\alphadot_1 = 4,\quad\deltadot_1 = 1,\quad \alphadot_2 = -2,\quad
\deltadot_2 = 5.
\label{adot_ddot_dat}
\end{equation}

In order to obtain rational values for the trigonometric
functions of $\alpha_j$, $\delta_j$ we set
\begin{equation}
\sigma_1 = \frac{1}{2},\quad \tau_1 = \frac{1}{3},\quad \sigma_2 =
\frac{2}{3},\quad \tau_2 = \frac{2}{5}
\label{sig_tau_dat}
\end{equation}
and define $\alpha_j$, $\delta_j$ $(j=1,2)$ by
\begin{equation*}
\cos\alpha_j = \frac{1-\sigma_j^2}{1+\sigma_j^2},\quad
\sin\alpha_j = \frac{2\sigma_j}{1+\sigma_j^2},\quad
\cos\delta_j = \frac{1-\tau_j^2}{1+\tau_j^2},\quad
\sin\delta_j = \frac{2\tau_j}{1+\tau_j^2}.
\end{equation*}

\end{appendix}

\bibliography{bib}{}
\bibliographystyle{plain}

\end{document}